\pdfoutput=1
\documentclass[11pt]{article}
\usepackage{ACL2023}
\usepackage{times}
\usepackage{latexsym}
\usepackage[T1]{fontenc}
\usepackage[utf8]{inputenc}
\usepackage{microtype}
\usepackage{inconsolata}

\usepackage{amsfonts}
\usepackage{amsmath}
\usepackage{amssymb}
\usepackage{amsthm}
\usepackage{enumitem}
\usepackage{thm-restate}
\usepackage{algorithm}
\usepackage{bbold}
\usepackage{caption}
\usepackage{subcaption}
\usepackage{setspace}
\usepackage[noend]{algpseudocode}
\usepackage{algorithm}
\usepackage{relsize}
\usepackage{hyperref}
\input{algonames.tex}

\usepackage{cleveref}
\crefname{section}{\S}{\S\S}
\crefname{table}{Tab.}{}
\crefname{figure}{Fig.}{}
\crefname{algorithm}{Alg.}{}
\crefname{equation}{Eq.}{}
\crefname{appendix}{App.}{}
\crefname{theorem}{Thm.}{}
\crefname{lemma}{Lem.}{}
\crefformat{section}{\S#2#1#3} 
\crefname{proposition}{Prop.}{}
\crefname{definition}{Def.}{}

\theoremstyle{plain}
\newtheorem{definition}{Definition}

\newtheorem{proposition}{Proposition}

\newcommand{\defn}[1]{\textbf{#1}}
\newcommand{\defeq}[0]{\mathrel{\stackrel{\textnormal{\tiny def}}{=}}}
\newcommand{\bigO}[1]{\mathcal{O}(#1)}
\newcommand{\kleene}[1]{#1^\ast}
\newcommand{\derives}{\overset{\ast}{\Rightarrow}}

\newcommand{\tree}{\boldsymbol{\tau}}
\newcommand{\trees}[2]{\mathcal{T}_{\nt{#1}}(#2)}
\newcommand{\str}{\boldsymbol{w}}
\newcommand{\valpha}{\boldsymbol{\alpha}}
\newcommand{\nt}[1]{\mathrm{#1}}
\newcommand{\alphabet}{\Sigma}
\newcommand{\nonterminals}{\mathcal{N}}
\newcommand{\start}{\nt{S}}
\newcommand{\rules}{\mathcal{R}}
\newcommand{\grammar}{\mathcal{G}}
\newcommand{\wcfgtuple}{(\nonterminals, \alphabet, \start, \rules, p)}
\newcommand{\pprefix}{\pi}
\newcommand{\ppre}[3]{\pprefix({#1},\nt{#3},{#2})}
\newcommand{\plc}{\xi}
\newcommand{\pinside}{\beta}
\newcommand{\pins}[3]{\pinside({#1},\nt{#3},{#2})}

\newcommand{\vzero}{\mathbf{0}}
\newcommand{\vone}{\mathbf{1}}
\newcommand{\seta}{\mathcal{A}}
\newcommand{\semiring}{\mathcal{W}}
\newcommand{\semiringtuple}{\langle \seta, \oplus, \otimes, \vzero, \vone \rangle}

\newcommand{\matA}{A}
\newcommand{\matB}{B}

\newcommand{\matP}{P}

\newcommand{\matI}{I}
\newcommand{\matM}{M}
\newcommand{\mplus}{+}

\newcommand{\mzero}{\mathbf{O}}
\newcommand{\mone}{\mathbf{I}}

\newcommand{\treesum}{\text{treesum}}

\newcommand{\scol}[1]{\color{black}{#1}} 
\newcommand{\xcol}[1]{\color{Orange}{#1}} 
\newcommand{\ycol}[1]{\color{LimeGreen}{#1}} 
\newcommand{\zcol}[1]{\color{Dandelion}{#1}} 
\newcommand{\jcol}[1]{\color{Plum}{#1}} 
\newcommand{\gcol}[1]{\color{red}{#1}} 
\newcommand{\dcol}[1]{\color{blue}{#1}}

\usepackage{todonotes}

\title{A Fast Algorithm for Computing Prefix Probabilities}

\author{Franz Nowak~\qquad~Ryan Cotterell \\
  \{\href{mailto:franz.nowak@inf.ethz.ch}{\texttt{franz.nowak}}\texttt{, }\href{mailto:ryan.cotterell@inf.ethz.ch}{\texttt{ryan.cotterell}}\}\texttt{@inf.ethz.ch}\\
  \fcolorbox{white}{white}{
  \includegraphics[width=.14\linewidth]{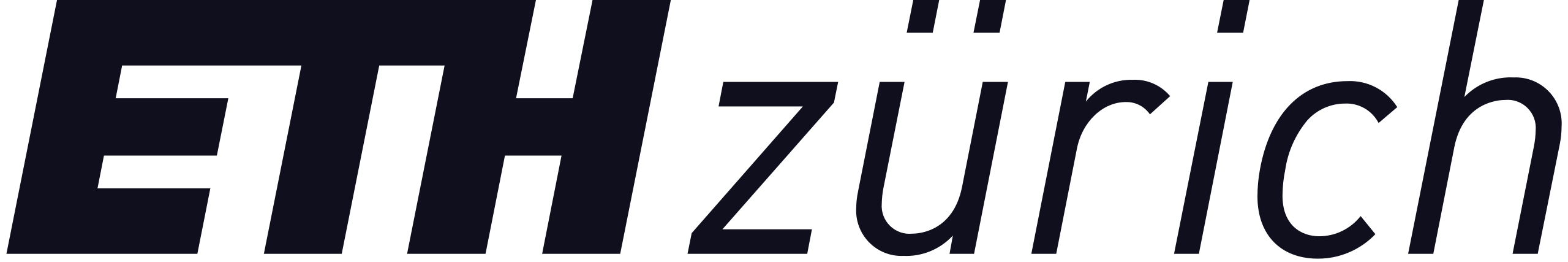}}
  }

\usepackage[pages=some]{background}
\usepackage{datetime}
\usepackage{etoolbox}
\usepackage{eso-pic}

\newcommand{\FirstPage}{57}
\newcommand{\LastPage}{67}
\newcommand{\Pages}[2]{\ifthenelse{\equal{#1}{#2}}{page #1}{pages #1--#2}}

\newcommand{\FirstLine}{Proceedings of the 61st Annual Meeting of the Association for Computational Linguistics
}
\newcommand{\SecondLine}{Volume 2: Short Papers,}
\newcommand{\ThirdLine}{July 9-14, 2023 \textcopyright2023 Association for Computational Linguistics}

\newcounter{Yloc}
\newcommand{\citeinfo}[2]{
  \AddToShipoutPicture*{%
    \setlength{\unitlength}{1mm}
    \ifdefempty{\FirstLine}{}{
    \put(105,\value{Yloc}){\makebox(0,0){\footnotesize { \FirstLine}}} %
    \addtocounter{Yloc}{-4}
    }
    \put(105,\value{Yloc}){\makebox(0,0){\footnotesize { \SecondLine} \em ~\Pages{\FirstPage}{\LastPage}}}
    \addtocounter{Yloc}{-4}
    \put(105,\value{Yloc}){\makebox(0,0){\footnotesize \em \ThirdLine}}
  }
}

\citeinfo{\FirstPage}{\LastPage}

\begin{document}
\maketitle
\begin{abstract}
Multiple algorithms are known for efficiently calculating the prefix probability of a string under a probabilistic context-free grammar (PCFG).
Good algorithms for the problem have a runtime cubic in the length of the input string.
However, some proposed algorithms are suboptimal with respect to the size of the grammar. 
This paper proposes a novel speed-up of \citeposs{jelinek-lafferty-1991-computation} algorithm,
whose original runtime is $\bigO{N^3 |\nonterminals|^3 + |\nonterminals|^4}$, where $N$ is the input length and $|\nonterminals|$ is the number of non-terminals in the grammar.
In contrast, our speed-up runs in $\bigO{N^2 |\nonterminals|^3+N^3|\nonterminals|^2}$.

\vspace{0.5em}
{\includegraphics[width=1.25em,height=1.25em]{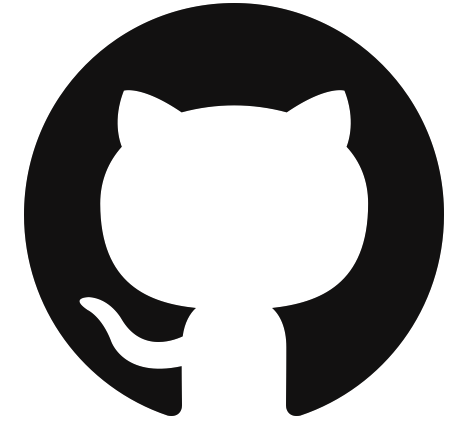}\hspace{1em}\parbox{\dimexpr\linewidth-2\fboxsep-2\fboxrule}{\url{https://github.com/rycolab/prefix-parsing}}}
\end{abstract}

\section{Introduction}
Probabilistic context-free grammars (PCFGs) are an important formalism in NLP \citep[Chapter 10] {eisenstein2019introduction}.
One common use of PCFGs is to construct a language model.
For instance, PCFGs form the backbone of many neural language models, e.g., recurrent neural network grammars \citep[RNNGs;][]{dyer-etal-2016-recurrent, dyer-2017-neural, kim-etal-2019-unsupervised}.
However, in order to use a PCFG as a language model, one needs to be able to compute prefix probabilities, i.e., the probability that the yield of a derivation starts with the given string.
In notation, given a string $\str = w_1 \cdots w_N$, 
we seek the probability $p(\nt{S} \derives \str \cdots)$ where $\nt{S}$ is the distinguished start symbol of the grammar and $\derives$ is the closure over applications of derivation rules of the grammar.\footnote{Specifically, $\boldsymbol{\alpha}\derives\boldsymbol{\beta}$ means that there exists an $n \geq 0$ such that $\boldsymbol{\alpha}\underbrace{\Rightarrow \cdots \Rightarrow}_{n \text{ times}}\boldsymbol{\beta}$, where $\Rightarrow$ marks a derivation step.}
Our paper gives a more efficient algorithm for the simultaneous computation of the prefix probabilities of \textit{all} prefixes of a string $\str$ under a PCFG.\looseness=-1

The authors are aware of two existing efficient algorithms to compute prefix probabilities under a PCFG.\footnote{Upon publication of this work, the authors were made aware of two other algorithms for finding prefix probabilities in the special case of idempotent semirings (\citealt{corazza-1994, sanchez-benedi-1997}). See \cref{sec:appendix-semirings} for a discussion of prefix parsing under a semiring.}
The first is \citeposs{jelinek-lafferty-1991-computation} algorithm which is derived from CKY \citep{Kasami1965AnER, YOUNGER1967189, Cocke1969ProgrammingLA}
and, thus, requires the grammar to be in Chomsky normal form (CNF).
Jelinek--Lafferty runs in $\bigO{N^3 |\nonterminals|^3 + |\nonterminals|^4}$ time, where $N$ is the length of the input and $\nonterminals$ is the number of non-terminals of the grammar, slower than the $\bigO{N^3 |\nonterminals|^3}$ required for parsing with CKY, when the number of non-terminals $|\nonterminals|$ is taken into account.\looseness=-1

The second, due to \citet{stolcke-1995-efficient}, is derived from Earley parsing \citep{earley-1970-efficient} and can parse arbitrary PCFGs,\footnote{Note that Earley's and, by extension, Stolcke's algorithms also implicitly binarize the grammar during execution by using dotted rules as additional non-terminals.} with a runtime of $\bigO{N^3|\nonterminals|^3}$.
Many previous authors have improved the runtime of Earley's \citep[\textit{inter alia}]{graham-1980, leermakers-1992, moore-2000-time}, and \citet{opedal-et-al-2023} successfully applied this speed-up to computing prefix probabilities, achieving a runtime of $\mathcal{O}(N^3|\grammar|)$, where $|\grammar|$ is the size of the grammar, that is, the sum of the number of symbols in all production rules.
\looseness=-1

Our paper provides a more efficient version of \citet{jelinek-lafferty-1991-computation} for the computation of prefix probabilities under a PCFG in CNF.
Specifically, we give an $\bigO{N^2 |\nonterminals|^3+N^3|\nonterminals|^2}$ time algorithm, which is the fastest attested in the literature for dense grammars in CNF,\footnote{A PCFG in CNF is dense if for every $\nt{X}, \nt{Y}, \nt{Z} \in \nonterminals$, we have a production rule $\nt{X} \rightarrow \nt{Y}\,\nt{Z} \in \rules$.} matching the complexity of CKY adapted for dense grammars by \citet{eisner-blatz-2007}.\footnote{Note that there exist approximate parsing algorithms with lower complexity bounds \citep{cohen-etal-2013-approximate}. Moreover, there are parsing algorithms that asymptotically run in sub-cubic time in the input length using fast matrix multiplication \citep{valiant-1975, benedi-et-al-2007}. However, they are of limited practical use \citep{lee-1997-fast}.}
We provide a full derivation and proof of correctness, as well as an open-source implementation on GitHub. We also briefly discuss how our improved algorithm can be extended to work for semiring-weighted CFGs.

\section{Preliminaries}
We start by introducing the necessary background on probabilistic context-free grammars.\looseness=-1
\begin{definition}
A \defn{probabilistic context-free grammar} (PCFG) is a five-tuple $\grammar=\wcfgtuple$, made up of:
\begin{itemize}
    \itemsep0em 
    \item A finite set of non-terminal symbols $\nonterminals$;
    \item A finite set of terminal symbols $\alphabet$, $\alphabet\cap\nonterminals{=}\emptyset$;
    \item A distinguished start symbol $\start\in\nonterminals$;
    \item A finite set of production rules $\rules \subset \nonterminals\times\kleene{\left(\nonterminals \cup \alphabet\right)}$ where each rule is written as $\nt{X}\xrightarrow{}\valpha$ with $\nt{X}\in\nonterminals$ and $\valpha\in\kleene{\left(\nonterminals \cup \alphabet\right)}$. Here, $\kleene{}$ denotes the Kleene closure;
    \item A weighting function $p \colon \rules \rightarrow [0, 1]$ assigning each rule $r\in\rules$ a probability such that $p$ is \defn{locally normalized}, meaning that for all $\nt{X}\in\nonterminals$ that appear on the left-hand side of a rule, $\sum\limits_{\nt{X}\xrightarrow{}\valpha\in\rules}p(\nt{X}\xrightarrow{}\valpha)=1$.
\end{itemize}
\end{definition}

Note that not every locally normalized PCFG constitutes a valid distribution over $\kleene{\alphabet}$.
Specifically, some may place probability mass on infinite trees \citep{chi-geman-1998-estimation}.
PCFGs that \emph{do} constitute a valid distribution over $\kleene{\alphabet}$ are referred to as \defn{tight}.
Furthermore, if all non-terminals of the grammar can be reached from the start non-terminal via production rules, we say the PCFG is \defn{trim}.

\begin{definition}
    A PCFG  $\grammar=\wcfgtuple$ is in \defn{Chomsky normal form} (CNF) if each production rule in $\rules$ is in one of the following forms:
\begin{align}
    &\nt{X}\xrightarrow{}\nt{Y}\,\nt{Z}\\
    &\nt{X}\xrightarrow{}a \\
    &\nt{S}\xrightarrow{}\varepsilon
\end{align}
where $\nt{X},\nt{Y},\nt{Z}\in\nonterminals$ such that $\nt{Y},\nt{Z} \neq \nt{S}$, $a\in\alphabet$, and $\varepsilon$ is the empty string.%
\footnote{Note that any PCFG can be converted to an equivalent PCFG in CNF \citep{smith-johnson-2007-weighted}.}
\end{definition}

\begin{definition}\label{def:derivation-step}
    A \defn{derivation step} $\boldsymbol{\alpha}\Rightarrow\boldsymbol{\beta}$ is an application of the binary relation $\Rightarrow:\kleene{(\nonterminals\cup\alphabet)}\times\kleene{(\nonterminals\cup\alphabet)}$, which rewrites the left-most non-terminal in $\boldsymbol{\alpha}$ according to a rule in $\rules$ from the left-hand side of that rule to its right-hand side, resulting in $\boldsymbol{\beta}$.
    The probability of a derivation step is the probability of the applied rule: $p(\boldsymbol{\alpha}\nt{X}\boldsymbol{\gamma}\Rightarrow\boldsymbol{\alpha}\boldsymbol{\beta}\boldsymbol{\gamma})\defeq p(\nt{X}\xrightarrow{}\boldsymbol{\beta})$.\looseness=-1
\end{definition}

\begin{definition}\label{def:derivation}
    A \defn{derivation} under a grammar $\grammar$ is a sequence $\boldsymbol{\alpha}_0, \boldsymbol{\alpha}_1, \cdots, \boldsymbol{\alpha}_m$, where $\boldsymbol{\alpha}_0\in\nonterminals, \boldsymbol{\alpha}_1, \cdots, \boldsymbol{\alpha}_{m-1} \in {\kleene{(\nonterminals\cup\alphabet)}}$, and $\boldsymbol{\alpha}_m \in \kleene{\alphabet}$, in which each $\boldsymbol{\alpha}_{i+1}$ is formed by applying a derivation step to $\boldsymbol{\alpha}_i$. $\boldsymbol{\alpha}_m=w_1 \cdots w_N\in\kleene{\Sigma}$ is called the \defn{yield} of the derivation. 
    If $\boldsymbol{\alpha}_0$ is not the start symbol $\nt{S}$, we call it a \defn{partial derivation}.
    We define $p(\boldsymbol{\alpha}\derives\boldsymbol{\beta})$ as the sum of probabilities of all subsequences from $\boldsymbol{\alpha}$ to $\boldsymbol{\beta}$, where each subsequence probability is the product of the individual derivation step probabilities defined in \Cref{def:derivation-step}.
\end{definition}
We represent derivations as trees whose structure corresponds to production rules, where any parent node is the non-terminal on the left-hand side of a rule and its children are the symbols from the right-hand side. The leaves of the tree, when read from left to right, form the yield. 
Such a tree, when rooted $\nt{S}$, is called a \defn{derivation tree}. Otherwise, it is called a \defn{derivation subtree}.
\begin{definition}
    The probability of a derivation tree (or derivation subtree) $\tree$ is the product of the probabilities of all its corresponding production rules:
    \begin{equation}
        p(\tree) \defeq \prod\limits_{(\boldsymbol{\alpha}\xrightarrow{}\boldsymbol{\beta})\in\tree}p(\boldsymbol{\alpha}\xrightarrow{}\boldsymbol{\beta})
    \end{equation}
\end{definition}
\vspace{-5pt}
\begin{definition}
$\trees{X}{w_i \cdots w_k}$ is the set of all derivation subtrees $\tree$ rooted at $\nt{X}$ with yield $w_i \cdots w_k$.\looseness=-1
\end{definition}

\begin{definition}
Given a PCFG $\grammar = \wcfgtuple$, a string $\str = w_1 \cdots w_N \in \kleene{\alphabet}$, and a non-terminal $\nt{X}\in\nonterminals$, the \defn{inside probability} of $\nt{X}$ between indices $i$ and $k$ (where $0 \leq i \leq k \leq N$) is defined as:\begin{align}
\pins{i}{k}{X} &\defeq p(\nt{X} \derives w_{i+1} \cdots w_k)\\
&= \sum_{\tree \in \trees{X}{w_{i+1} \cdots w_k}} p(\tree)
\end{align}
That is, the sum of the probabilities of all derivation trees $\tree$ starting at $X$ that have yield $w_i \cdots w_k$. 
\end{definition}

\begin{definition}
Given a PCFG $\grammar = \wcfgtuple$, a string $\str = w_1 \cdots w_N \in \kleene{\alphabet}$, and a non-terminal $\nt{X}\in\nonterminals$, we define the \defn{prefix probability} $\pprefix$, i.e., the probability of $\str$ being a prefix under $\grammar$, as:
\begin{align}
    \pprefix(\str\mid\nt{X}) \defeq \sum_{\boldsymbol{u} \in \kleene{\alphabet}}p(\nt{X} \derives \str \boldsymbol{u})    
\end{align}
\end{definition}
\vspace{-3pt}
In words, $\pprefix$ is the probability of deriving $\str$ with an arbitrary continuation from $\nt{X}$, that is, the sum of probabilities of deriving $\str\boldsymbol{u}$ from $\nt{X}$ over all possible suffixes $\boldsymbol{u}\in\kleene{\alphabet}$. 
In the following, we write the prefix probability of deriving prefix $\str = w_{i+1} \cdots w_k$ from $\nt{X}$ as $\ppre{i}{k}{X}$.

\begin{figure*}[ht!]
\centering
\begin{minipage}{0.48\textwidth}
\begin{algorithm}[H]
\caption{CKY}
\setstretch{0.86}
\label{alg:cky}
\begin{algorithmic}[1]
    \Func{CKY($\str=w_1 \cdots w_N$)}
    \LineComment{Initialize inside probabilities}
    \State $\pins{\cdot}{\cdot}{\cdot} \xleftarrow{} 0$
    \LineComment{Handle special rule, $\nt{S} \xrightarrow{} \varepsilon$}
    \State $\pins{0}{0}{\nt{S}} \xleftarrow{} p(\nt{S} \xrightarrow{} \varepsilon)$
    \For{$k\in 0,\dots,N-1$}
        \For{$\nt{X}\xrightarrow{}w_{k+1} \in \rules$}
            \LineComment{Handle single word tokens}
            \State $\pins{k}{k{+}1}{X}$ \texttt{+=} $p(\nt{X}\xrightarrow{}w_{k+1})$
        \EndFor
    \EndFor
    \LineComment{$\ell$ is the span size}
    \For{$\ell\in2,\dots,N$}\label{cky:main-loop}
        \LineComment{$i$ marks the beginning of the span}
        \For{$i\in0,\dots,N-\ell$}
            \LineComment{$k$ marks the end of the span}
            \State $k \xleftarrow{} i+\ell$
            \LineComment{Recursively compute $\pinside$}
            \For{$\nt{X}\xrightarrow{}\nt{Y\ Z}\in\rules$}
                \State $\pins{i}{k}{X}$ \texttt{+=} $p(\nt{X}{\xrightarrow{}}\nt{Y\,Z})$
                \Statex \quad\quad\quad\quad\quad $\cdot\sum\limits_{j=i+1}^{k-1}\pins{i}{j}{Y}\cdot\pins{j}{k}{Z}$
            \EndFor
        \EndFor
    \EndFor\label{cky:main-loop-end}
    \State \Return $\pinside$
\EndFunc
\end{algorithmic}
\end{algorithm}    
\end{minipage}
\hfill
\begin{minipage}{.5\textwidth}
\begin{algorithm}[H]
\caption{Jelinek--Lafferty}
\setstretch{0.77} 
\label{alg:lri}
\begin{algorithmic}[1]
    \State $\matP' \xleftarrow{}\left(\matI - \matP\right)^{-1}$ \InlineComment{Precompute $\kleene{\matP}$ by \cref{eq:pl}}\label{jl:pstar}
    \For {$\nt{X}_i,\nt{X}_j\in\nonterminals$} \InlineComment{Assign $\plc(\nt{Y}\mid X)$}\label{jl:init-xi}
            \State $\plc(\nt{X}_j\mid \nt{X}_i)\xleftarrow{} \matP'_{ij}$
    \EndFor\label{jl:init-xi-end}
        \For {$\nt{X'}\xrightarrow{}\nt{Y\ Z}\in\rules$} \InlineComment{Precompute $\plc(\nt{Y\,Z}\mid X)$}\label{jl:unfold-xi}
        \State $\plc(\nt{Y}\,\nt{Z} \mid \nt{X}){\xleftarrow{}}\sum\limits_{\nt{X}\in\nonterminals} \plc(\nt{X'} \mid \nt{X}) \cdot p(\nt{X'}{\xrightarrow{}}\nt{Y}\,\nt{Z}) $
    \EndFor\label{jl:unfold-xi-end}
    \Func{JL($\str=w_1 \cdots w_N$)}
    \State $\ppre{\cdot}{\cdot}{\cdot} \xleftarrow{} 0$ \InlineComment{Initialize prefix probabilities}\label{jl:init-pi}
    \For{$k \in 0,\dots, N$}
        \For{$\nt{X} \in \nonterminals$} \InlineComment{Prefix probability of $\varepsilon$}
            \State $\ppre{k}{k}{X}\xleftarrow{}1$
        \EndFor
    \EndFor\label{jl:end-init-pi}
    \State $\pinside \xleftarrow{} \text{CKY}(\str)$ \InlineComment{Compute $\pinside$ with \cref{alg:cky}}\label{jl:beta}
    \For{$k \in 0,\dots, N-1$}\label{jl:base-case}
        \For{$\nt{X} \in \nonterminals$} \InlineComment{Compute base case}
            \State $\pprefix(k{,}\nt{X}{,}k{+}1){\xleftarrow{}}\sum\limits_{\nt{Y}{\in}\nonterminals}\plc(\nt{Y} {\mid} \nt{X}){\cdot} p(\nt{Y}{\xrightarrow{}}w_{k+1})$
        \EndFor
    \EndFor\label{jl:base-case-end}
    \For{$\ell \in 2 \hdots N$}\label{jl:main-loop}
        \For{$i \in 0 \hdots N-\ell$}
            \State $k\xleftarrow{}i+\ell$
            \For{$\nt{X,Y,Z}\in\nonterminals$} \InlineComment{Recursively compute $\pprefix$}
                \State $\ppre{i}{k}{X}$ \texttt{+=} $\plc(\nt{Y\ Z} \mid \nt{X})$
                \Statex \quad\quad\quad\quad\quad $\cdot\sum\limits_{j=i+1}^{k-1}\pins{i}{j}{Y}\cdot \ppre{j}{k}{Z}$
            \EndFor
        \EndFor
    \EndFor\label{jl:main-loop-end}
    \State\Return $\pprefix$
\EndFunc
\end{algorithmic}
\end{algorithm}
\end{minipage}
\caption{Pseudocode for the CKY algorithm (left) and Jelinek--Lafferty (right)}
\end{figure*}

\begin{definition}
Let $\grammar$ be a PCFG in CNF. Then for non-terminals $\nt{X}, \nt{Y}, \nt{Z} \in \nonterminals$, the \defn{left-corner} expectations $\plc(\nt{Y} \mid \nt{X})$ and $\plc(\nt{Y}\,\nt{Z} \mid \nt{X})$ are defined as: 
\begin{align}
\plc(\nt{Y} \mid \nt{X}) &\defeq \sum_{\boldsymbol{u}\in\kleene{\alphabet}} p(\nt{X} \derives \nt{Y} \boldsymbol{u})\label{eq:plc_1} \\
\plc(\nt{Y}\,\nt{Z} \mid \nt{X}) &\defeq \sum_{\nt{X'}\in\nonterminals} \plc(\nt{X'} \mid \nt{X}) \cdot p(\nt{X'}{\xrightarrow{}}\nt{Y}\,\nt{Z})  \label{eq:plc_2}
\end{align}
\end{definition}

\usetikzlibrary{decorations.pathmorphing}
\begin{figure}[H]{}
    \centering
    \begin{subfigure}{0.24\textwidth}
        \centering
        \begin{tikzpicture}[]
            \node(0){X}
                child[missing]{}
                child{node{$\cdots$}};
            \begin{scope}[yshift=-0.8in,xshift=-0.4in]
            \node(1){Y};
            \end{scope}
            \draw [->,decorate,decoration={snake,amplitude=.4mm,segment length=2mm,post length=1mm}]
            (0) -- (1);
            \begin{scope}[yshift=-0.9in,xshift=-0.4in]
            \draw (0,0) node{}
              -- (-0.5,-0.5) node{}
              -- (0.5,-0.5) node{}
              -- cycle;
            \end{scope}
        \end{tikzpicture}    
        \caption{\normalsize$\plc(\nt{Y}\mid\nt{X})$}
        \label{fig:leftcorner1}
    \end{subfigure}
    \begin{subfigure}{0.23\textwidth}
    \centering
        \begin{tikzpicture}[]
            \node(0){X}
                child[missing]{}
                child{node{$\cdots$}};
            \begin{scope}[yshift=-0.45in,xshift=-0.23in]
            \node(1){$\nt{X'}$}
            child{node(2){Y}}
            child{node(3){Z}};
            \end{scope}
            \draw [->,decorate,decoration={snake,amplitude=.4mm,segment length=2mm,post length=1mm}]
            (0) -- (1);
            \begin{scope}[yshift=-1.15in,xshift=-0.52in]
            \draw (0,0) node{}
              -- (-0.5,-0.5) node{}
              -- (0.5,-0.5) node{}
              -- cycle;
            \end{scope}
            \begin{scope}[yshift=-1.15in,xshift=0.05in]
            \draw (0,0) node{}
              -- (-0.5,-0.5) node{}
              -- (0.5,-0.5) node{}
              -- cycle;
            \end{scope}
        \end{tikzpicture}    
        \caption{\normalsize$\plc(\nt{Y}\,\nt{Z}\mid\nt{X})$}
        \label{fig:leftcorner2}
    \end{subfigure}
    \caption{Visualization of left-corner expectations}
\end{figure}
\vspace{-11pt}
The left-corner expectation $\plc(\nt{Y}\mid\nt{X})$ is hence the sum of the probabilities of partial derivation subtrees rooted in $\nt{X}$ that have $\nt{Y}$ as the left-most leaf; see \cref{fig:leftcorner1} for a visualization. 
Similarly, $\plc(\nt{Y}\,\nt{Z}\mid\nt{X})$ is the sum of the probabilities of partial derivation subtrees that have $\nt{Y}$ and $\nt{Z}$ as the leftmost leaves; see \cref{fig:leftcorner2}.\looseness=-1

\section{\citet{jelinek-lafferty-1991-computation}}
We now give a derivation of the Jelinek--Lafferty algorithm.
The first step is to derive an expression for the prefix probability in PCFG terms.
\begin{restatable}{lemma}{recursion}
    \label{lem:ppre}
    Given a tight, trim PCFG in CNF and a string $\str = w_1 \cdots w_N$, the prefix probability of a substring $w_{i+1} \cdots w_k$ of $\str$ being derived from $\nt{X}$ can be defined recursively as follows:
\end{restatable}
\vspace{-11pt}
\begin{align}
    \begin{split}
        \ppre{i}{k}{X} &= \sum_{\nt{Y},\nt{Z}\in\nonterminals}\plc(\nt{Y}\,\nt{Z}\mid \nt{X})\\
      \quad  \cdot\,\sum_{j=i+1}^{k-1} &\pins{i}{j}{Y}\cdot\ppre{j}{k}{Z} \label{eq:ppre}
    \end{split}
\end{align}
\textit{Base case (for all $k\in0\dots N-1$ and all $\nt{X}\in\nonterminals$):}\looseness=-1
\begin{align}
    \ppre{k}{k{+}1}{X} = \sum\limits_{\nt{Y}{\in}\nonterminals}\plc(\nt{Y} {\mid} \nt{X}) {\cdot} p(\nt{Y}{\xrightarrow{}}w_{k+1}) \label{eq:base}
\end{align}
\noindent
\emph{Proof.} A proof of \Cref{lem:ppre} is given in \cref{sec:appendix}.\vspace{5pt}
\noindent
The above formulation of the prefix probability is closely related to that of the inside probability from \citeposs{Baker79} inside--outside algorithm, which can be efficiently computed using CKY, see \cref{alg:cky}.\looseness=-1

\noindent
Next, the left-corner expectations $\plc$ as defined by \cref{eq:plc_1} can be computed efficiently as follows.
Let $\matP$ denote the square matrix of dimension $|\nonterminals|$, with rows and columns indexed by the non-terminals $\nonterminals$ (in some fixed order), where the entry at the $i^{\text{th}}$ row and the $j^{\text{th}}$ column corresponds to $p(\nt{X}_i\xrightarrow{}\nt{X}_j\ \_)$, i.e., the probability of deriving $\nt{X}_j$ on the left corner from $\nt{X}_i$ in one step (we use $\_$ as a wildcard):
\begin{equation}
    p(\nt{X}_i\xrightarrow{}\nt{X}_j\ \_) \defeq \sum_{\nt{Y}\in\nonterminals} p(\nt{X}_i\xrightarrow{}\nt{X}_j\ \nt{Y})
\end{equation}
We can find the probability of getting to non-terminal $\nt{X}_j$ after $k$ derivation steps starting from $\nt{X}_i$ by multiplying $\matP$ with itself $k$ times:
\begin{equation}
p(\nt{X}_i\xrightarrow{k}\nt{X}_j\ \_)=(\matP^k)_{ij}
\end{equation}
We can hence get the matrix $\kleene{\matP}$, whose entries correspond to deriving $\nt{X}_j$ from $\nt{X}_i$ after \emph{any} number of derivation steps, by summing over all the powers of the matrix $\matP$:\footnote{Note that this sum converges if the PCFG is tight and trim since infinite derivation (sub)trees have zero probability mass.}
\begin{align}
    \kleene{\matP} &\defeq \matI + \matP + \matP^2 + \matP^3 + \cdots = \sum_{n=0}^\infty \matP^n \label{eq:pl}\\
    &= \matI + \matP \sum_{n=0}^\infty \matP^n = \matI + \matP\kleene{\matP} = (\matI - \matP)^{-1} \nonumber
\end{align}
Note that the entry at the $i^{\text{th}}$ row and $j^{\text{th}}$ column of $\kleene{\matP}$ is exactly the left-corner expectation $\plc(\nt{X}_j \mid \nt{X}_i)$.
Finally, we can compute the left-corner expectations $\plc(\nt{Y}\,\nt{Z}\mid\nt{X})$ using \cref{eq:plc_2}:
\begin{equation*}
\plc(\nt{Y}\,\nt{Z} \mid \nt{X}) \defeq \sum_{\nt{X'}\in\nonterminals} \plc(\nt{X'} \mid \nt{X}) \cdot p(\nt{X'}{\xrightarrow{}}\nt{Y}\,\nt{Z})    
\end{equation*}
Lastly, for completeness, we also compute the prefix probability of the empty string, $\varepsilon$. 
For probabilistic PCFGs, this probability is simply 1 because $\varepsilon$ is a prefix of any string:%
\footnote{For a generalization of the algorithm that does not require locally normalized rule weights, see \cref{sec:appendix-semirings}.}
\begin{align}
    \ppre{k}{k}{X} \defeq \pprefix(\varepsilon \mid \nt{X}) =  1
\end{align}

We can now combine the quantities derived above to obtain an efficient algorithm for the computation of prefix probabilities $\ppre{i}{k}{S}$.
For the full algorithm, see \cref{alg:lri}.

\begin{figure*}
\begin{subequations}
\begin{align}
\ppre{i}{k}{\scol{X}} &=  \sum_{\nt{\ycol{Y}},\nt{\zcol{Z}}\in\nonterminals}\plc(\nt{\ycol{Y}},\nt{\zcol{Z}}\mid \nt{\scol{X}})\cdot\sum_{{\jcol{j}}=i+1}^{k-1} \pins{i}{{\jcol{j}}}{\ycol{Y}}\cdot\ppre{{\jcol{j}}}{k}{\zcol{Z}}\label{eq:defn}\\
&=\sum_{\nt{\ycol{Y}},\nt{\zcol{Z}}\in\nonterminals}\sum_{\nt{\xcol{X'}}\in\nonterminals}
\plc(\nt{\xcol{X'}}\mid\nt{\scol{X}})\cdot p(\nt{\xcol{X'}}\xrightarrow{}\nt{\ycol{Y}}\,\nt{\zcol{Z}})
\cdot\sum_{{\jcol{j}}=i+1}^{k-1}\pins{i}{{\jcol{j}}}{\ycol{Y}}\cdot\ppre{{\jcol{j}}}{k}{\zcol{Z}}\label{eq:expand}\\
&=\sum_{\nt{\xcol{X'}},\nt{\zcol{Z}}\in\nonterminals}\plc(\nt{\xcol{X'}}\mid\nt{\scol{X}})\cdot \sum_{{\jcol{j}}=i+1}^{k-1}{\gcol{\gamma}}_{i{\jcol{j}}}(\nt{\xcol{X'}},\nt{\zcol{Z}})\cdot\ppre{{\jcol{j}}}{k}{\zcol{Z}}\label{eq:factor}\\
&= \sum_{\nt{\zcol{Z}}\in\nonterminals}\sum_{{\jcol{j}}=i+1}^{k-1}  {\dcol{\delta}}_{i{\jcol{j}}}(\nt{{\scol{X}}},\nt{\zcol{Z}})\cdot\ppre{{\jcol{j}}}{k}{\zcol{Z}}
\label{eq:fast_ppre}\\
\text{where } {\gcol{\gamma}}_{i{\jcol{j}}} (\nt{\xcol{X'}}&,\nt{\zcol{Z}})\defeq\sum_{\nt{\ycol{Y}}\in\nonterminals}p(\nt{\xcol{X'}}\xrightarrow{}\nt{\ycol{Y}}\,\nt{\zcol{Z}})\cdot \pins{i}{\jcol{j}}{\ycol{Y}}\label{eq:gamma}\\
    \text{ and } {\dcol{\delta}}_{i{\jcol{j}}}(\nt{{\scol{X}}}&,\nt{\zcol{Z}})\defeq\sum_{\nt{\xcol{X'}}\in\nonterminals}\plc(\nt{\xcol{X'}}\mid\nt{\scol{X}}) \cdot {\gcol{\gamma}}_{i{\jcol{j}}}(\nt{\xcol{X'}},\nt{\zcol{Z}})\label{eq:delta}
\end{align}
\end{subequations}
\label{fig:speedup}
\end{figure*}

\begin{proposition}
    The time complexity of the CKY algorithm as presented in \cref{alg:cky} is $\bigO{N^3|\nonterminals|^3}$. \label{prop:cky}
\end{proposition}
\begin{proof}
Clearly, the computationally critical part is in \crefrange{cky:main-loop}{cky:main-loop-end}, where we iterate over all indices of $\str$ for $i$, $j$, and $k$, as well as over the whole set of grammar rules, thus taking $\bigO{N^3|\rules|}$.  In a PCFG in CNF, with the size of $\alphabet$ taken as constant, the number of rules, $|\rules|$, is $\bigO{|\nonterminals|^3}$, making the overall complexity of CKY $\bigO{N^3|\nonterminals|^3}$.
\end{proof}
\begin{proposition}
The total time complexity of Jelinek--Lafferty is $\bigO{N^3|\nonterminals|^3 + |\nonterminals|^4}$:
\end{proposition}
\begin{proof}
First, we precompute any values that are independent of the input. In \crefrange{jl:pstar}{jl:init-xi-end}, we precompute all the left-corner expectations $\plc(\nt{Y}\mid\nt{X})$ using \cref{eq:pl}, which has the complexity of inverting the matrix $\matP$, i.e., $\bigO{|\nonterminals|^3}$, and move the values into a map of left-corner expectations in $\bigO{|\nonterminals|^2}$ (this is just for readability). 
In \crefrange{jl:unfold-xi}{jl:unfold-xi-end}, we then use \cref{eq:plc_2} to compute $\plc(\nt{Y}\,\nt{Z}\mid\nt{X})$, iterating once over all non-terminals $\nt{X}$ for each rule, which takes $\bigO{|\rules||\nonterminals|}$, that is, $\bigO{|\nonterminals|^4}$.
After initializing the probabilities in \crefrange{jl:init-pi}{jl:end-init-pi} in $\bigO{N^2|\nonterminals|}$, we begin by pre-computing all the inside probabilities $\pinside$ in \cref{jl:beta} of \cref{alg:lri}, which takes $\bigO{N^3|\nonterminals|^3}$ by \cref{prop:cky}.
Computing $\ppre{k}{k+1}{X}$ for all $\nt{X}\in\nonterminals$ by \cref{eq:base} in \crefrange{jl:base-case}{jl:base-case-end} takes $\bigO{N|\nonterminals|^2}$ as we iterate over all positions $k\in N$ and over all $\nt{Y}\in\nonterminals$ for each $\nt{X}\in\nonterminals$.
And finally, computing the $\pprefix$ chart in \crefrange{jl:main-loop}{jl:main-loop-end} takes $\bigO{N^3|\nonterminals|^3}$ since we iterate over all $\ell, i, j \leq N$ and $\nt{X}, \nt{Y}, \nt{Z} \in \nonterminals$. 
This yields an overall time complexity of $\bigO{N^3|\nonterminals|^3 + |\nonterminals|^4}$.
\end{proof}

\section{Our Speed-up}
We now turn to our development of a faster dynamic program to compute all prefix probabilities.

\noindent
The speed-up comes from a different way to factorize $\ppre{i}{k}{X}$, which allows additional memoization. 
Starting with the definition of the prefix probability in \cref{eq:defn}, we first expand $\plc(\nt{Y}\,\nt{Z}\mid\nt{X})$ by \cref{eq:plc_2}, as seen in \cref{eq:expand}. 
Then, we factor out all terms that depend on the left-corner non-terminal $\nt{Y}$ in \cref{eq:factor}, which we store in a chart $\gamma$, see \cref{eq:gamma}. We then do the same for all terms depending on $\nt{X'}$, factoring them out in \cref{eq:fast_ppre} and storing them in another chart $\delta$, see \cref{eq:delta}.\looseness=-1

Our improved algorithm for computing all prefix probabilities is shown in \cref{alg:fast_lri}.

\begin{algorithm}[ht!]
\caption{Faster prefix probability algorithm}
\setstretch{0.9}
\label{alg:fast_lri}
\begin{algorithmic}[1]
    \State $\matP' \xleftarrow{} (\matI - \matP)^{-1}$ \InlineComment{Precompute $\kleene{\matP}$ with \cref{eq:pl}}\label{fjl:pstar}
    \For {$\nt{X}_i,\nt{X}_j\in\nonterminals$} \InlineComment{Assign $\plc(\nt{Y}\mid X)$}\label{fjl:init-xi}
            \State $\plc(\nt{X}_j\mid \nt{X}_i)\xleftarrow{} \matP'_{ij}$
    \EndFor\label{fjl:init-xi-end}
    \Func{FastJL($\str=w_1 \cdots w_N$)}
    \State $\ppre{\cdot}{\cdot}{\cdot} \xleftarrow{} 0$ \InlineComment{Initialize prefix probabilities}\label{fjl:init-pi}
    \For{$k \in 0,\dots, N$}
        \For{$\nt{X} \in \nonterminals$} \InlineComment{Prefix probability of $\varepsilon$}
            \State $\ppre{k}{k}{X}\xleftarrow{}1$
        \EndFor
    \EndFor\label{fjl:init-pi-end}
    \State $\pinside \xleftarrow{} \text{CKY}(\str)$\InlineComment{Compute $\pinside$ with \cref{alg:cky}}\label{fjl:beta}
    \For{$i,j=0,\dots, N$}\label{fjl:gamma-delta}
        \For{$\nt{X,Z}\in\nonterminals$}\InlineComment{Compute $\gamma$ by \cref{eq:gamma}}
            \State $\gamma_{ij}(\nt{X{,}Z}){\xleftarrow{}}\sum\limits_{\nt{Y}{\in}\nonterminals} p(\nt{X}{\xrightarrow{}}\nt{Y}\nt{Z}){\cdot}\pins{i}{j}{Y}$
        \EndFor
        \For{$\nt{X,Z}\in\nonterminals$}\InlineComment{Compute $\delta$ by \cref{eq:delta}}
            \State $\delta_{ij}(\nt{X{,}Z}){\xleftarrow{}}\sum\limits_{\nt{Y}\in\nonterminals}\plc(\nt{Y}\mid\nt{X})\cdot\gamma_{ij}(\nt{Y},\nt{Z})$
        \EndFor
    \EndFor\label{fjl:gamma-delta-end}
    \For{$k \in 0,\dots, N-1$}\label{fjl:base-case}
        \For {$\nt{X} \in \nonterminals$} \InlineComment{Compute base case}
            \State $\pprefix(k ,\nt{X}, k{+}1) \xleftarrow{} \sum\limits_{\nt{Y}\in\nonterminals}\plc(\nt{Y}\mid\nt{X})$
            \Statex \quad \quad \quad \quad $\cdot p(\nt{Y}\xrightarrow{}w_{k{+}1})$
        \EndFor
    \EndFor\label{fjl:base-case-end}
    \For{$\ell \in 2 \hdots N$}\label{fjl:main-loop}
        \For{$i \in 1 \hdots N-\ell$}
            \State $k\xleftarrow{}i+\ell$
            \For{$\nt{X, Z}\in\nonterminals$} \InlineComment{Recursively compute $\pprefix$}
                \State $\ppre{i}{k}{X}{\texttt{+=}}\sum\limits_{j={i+1}}^{k-1}\delta_{ij}( \nt{X},\nt{Z}){\cdot}\ppre{j}{k}{Z}$
            \EndFor
        \EndFor
    \EndFor\label{fjl:main-loop-end}
    \State\Return $\pprefix$
\EndFunc
\end{algorithmic}
\end{algorithm}

\begin{proposition}
    The complexity of our improved algorithm is $\bigO{N^2|\nonterminals|^3+N^3|\nonterminals|^2}$.
\end{proposition}
\begin{proof}
As before, \cref{alg:fast_lri} starts by precomputing and assigning the left-corner expectations in \crefrange{fjl:pstar}{fjl:init-xi-end}, which takes $\bigO{|\nonterminals|^3}$ and $\bigO{|\nonterminals|^2}$, respectively.
We then initialize the prefix probabilities and compute the inside probabilities in \crefrange{fjl:init-pi}{fjl:init-pi-end}, taking $\bigO{N^2|\nonterminals|}$.
As \citet{eisner-blatz-2007} show, one can compute $\pinside$ (\cref{fjl:beta}) in $\bigO{N^2|\nonterminals|^3 + N^3|\nonterminals|^2}$, thus improving the runtime of \cref{alg:cky} for dense grammars.
Pre-computing $\gamma$ and $\delta$ in \crefrange{fjl:gamma-delta}{fjl:gamma-delta-end} takes $\bigO{N^2|\nonterminals|^3}$, as we sum over non-terminals, and both charts each have two dimensions indexing $N$ and two indexing $\nonterminals$.
Computing the base case $\ppre{k}{k+1}{X}$ for all non-terminals $\nt{X}$ and positions $k$ in \crefrange{fjl:base-case}{fjl:base-case-end} takes $\bigO{N|\nonterminals|^2}$, as before.
Finally, the loops computing $\pprefix$ in \crefrange{fjl:main-loop}{fjl:main-loop-end} take $\bigO{N^3|\nonterminals|^2}$, as we are now iterating over $\nt{X},\nt{Z}\in\nonterminals$ and $\ell, i, j\leq N$.
Hence, our new overall time complexity is $\bigO{N^2|\nonterminals|^3+N^3|\nonterminals|^2}$.
\end{proof}

\section{Generalization to Semirings}
It turns out that Jelinek--Lafferty, and, by extension, our improved algorithm, can be generalized to work for semiring-weighted CFGs, 
with the same time complexity, under the condition that 
the semiring is closed, i.e., it has a Kleene star. 
This follows from the fact that the only operations used by the algorithm are addition and multiplication if we use \citeposs{lehmann-1977-algebraic} algorithm for the computation of left-corner expectations, $\plc$.
We can even relax the condition of local normalization by changing how left-corner expectations and the weight of the prefix $\varepsilon$ are computed.
The relevant definitions and derivation of the adapted algorithm can be found in \cref{sec:appendix-semirings}.\looseness=-1

\section{Conclusion}
In this paper, we have shown how to efficiently compute prefix probabilities for PCFGs in CNF, adapting Jelinek--Lafferty to use additional memoization, thereby reducing the time complexity from $\bigO{N^3|\nonterminals|^3 + |\nonterminals|^4}$ to $\bigO{N^2|\nonterminals|^3 + N^3|\nonterminals|^2}$. We thereby addressed one of the main limitations of the original formulation, of being slow for large grammar sizes.\looseness=-1

\section*{Limitations}
While we have improved the asymptotic running time of a classic algorithm with regard to grammar size, the time complexity of our algorithm is still cubic in the length of the input. Our result follows the tradition of dynamic programming algorithms that trade time for space by memoizing and reusing precomputed intermediate results. The usefulness of this trade-off in practice depends on the specifics of the grammar, and while the complexity is strictly better in terms of non-terminals, it will be most noticeable for denser grammars with many non-terminals.

\section*{Ethics Statement}
We do not foresee any ethical issues arising from this work.

\section*{Acknowledgements}
We thank the anonymous reviewers for their helpful comments and suggestions. 
We also extend our gratitude to Abra Ganz, Andreas Opedal, Anej Svete, and Tim Vieira for their valuable feedback on various versions of this paper.

\bibliography{anthology,custom}
\bibliographystyle{acl_natbib}

\onecolumn
\appendix

\section{Proof of \cref{lem:ppre}} \label{sec:appendix}
\recursion*
\vspace{-11pt}
\setcounter{equation}{9}
\begin{align}
    \ppre{i}{k}{X} = \sum_{\nt{Y},\nt{Z}\in\nonterminals}\plc(\nt{Y}\,\nt{Z}\mid \nt{X}) \cdot\sum_{j=i+1}^{k-1} \pins{i}{j}{Y}\cdot\ppre{j}{k}{Z}
\end{align}
\textit{Base case (for all $k\in0\dots N-1$ and all $\nt{X}\in\nonterminals$):}
\begin{align}
    \ppre{k}{k{+}1}{X} = \sum\limits_{\nt{Y}{\in}\nonterminals}\plc(\nt{Y} \mid \nt{X}) \cdot p(\nt{Y}{\xrightarrow{}}w_{k+1})
\end{align}
\setcounter{equation}{16}
\begin{proof}
\cref{eq:base}: The base case, $\ppre{k}{k{+}1}{X}$, is the probability of $w_{k+1}$ being the left-most terminal in the parse subtree under $\nt{X}$. 
It is, therefore, simply the sum of probabilities of any non-terminal $\nt{Y}$ being on the left corner of the parse subtree under $\nt{X}$ multiplied by the corresponding probability of $\nt{Y}$ directly deriving $w_{k+1}$.\looseness=-1

\cref{eq:ppre}:
Given the PCFG is in CNF and assuming $k>i+1$, in order to derive the prefix $w_{i+1} \cdots w_k$ we must first apply some rule $\nt{X}\xrightarrow{}\nt{Y}\,\nt{Z}$, where the first part of the substring is then derived from $\nt{Y}$ and the remainder (and potentially more) from $\nt{Z}$:
\begin{equation}
    \ppre{i}{k}{X} = \sum_{\nt{Y},\nt{Z}\in\nonterminals}p(\nt{X}\xrightarrow{}\nt{Y}\,\nt{Z})\, \left[\sum_{j=i+1}^{k-1}\pins{i}{j}{Y}\cdot\ppre{j}{k}{Z}+\ppre{i}{k}{Y}\right] \label{eq:ppre1}
\end{equation}
where the last term, $\ppre{i}{k}{Y}$, handles the case where the whole prefix is derived from $\nt{Y}$ alone. This term is clearly recursively defined through \cref{eq:ppre1}, with $\nt{X}$ replaced by $\nt{Y}$. 
Defining $R(\nt{Y},\nt{Z})\defeq\sum_{j=i+1}^{k-1}\pins{i}{j}{Y}\ppre{j}{k}{Z}$, we can rewrite \cref{eq:ppre1} as:
\begin{equation}
        \ppre{i}{k}{X}=\sum_{\nt{Y},\nt{Z}\in\nonterminals}p(\nt{X}\xrightarrow{}\nt{Y}\,\nt{Z})\cdot R(\nt{Y},\nt{Z})+\,\sum_{\nt{A},\nt{B}\in\nonterminals}p(\nt{X}\xrightarrow{}\nt{A}\,\nt{B})\cdot\ppre{i}{k}{A}
\end{equation}
After repeated substitutions ad infinitum, we get:
\begin{equation}
    \ppre{i}{k}{X}=\sum_{\nt{A},\nt{B}\in\nonterminals}p(\nt{X}\derives\nt{A}\,\nt{B})\,\sum_{\nt{Y},\nt{Z}\in\nonterminals}p(\nt{A}\xrightarrow{}\nt{Y}\,\nt{Z})\cdot R(\nt{Y},\nt{Z})
\end{equation}
Note that, in the last step, infinite derivations do not carry any probability mass since we assumed the PCFG to be tight and trim.
Hence, the final form of the equation is:
\begin{align}
    \begin{split}
        \ppre{i}{k}{X}&=\sum_{\nt{A},\nt{B}\in\nonterminals}p(\nt{X}\derives\nt{A}\,\nt{B})\sum_{\nt{Y},\nt{Z}\in\nonterminals}p(\nt{A}\xrightarrow{}\nt{Y}\,\nt{Z})\cdot R(\nt{Y},\nt{Z})\\
        &=\sum_{\nt{Y},\nt{Z}\in\nonterminals}\plc(\nt{Y}\,\nt{Z}\mid \nt{X})\sum_{j=i+1}^{k-1}\pins{i}{j}{Y}\cdot\ppre{j}{k}{Z}
    \end{split}
\end{align}
\end{proof}

\newpage
\section{Extension of \cref{alg:fast_lri} to Semirings} \label{sec:appendix-semirings}
In the following, we give the necessary background on semirings and then show how the algorithms introduced above can be framed in terms of semirings.

\subsection{Semirings}
We start by introducing the necessary definitions and notation.
\begin{definition}
    A \defn{monoid} is a 3-tuple $\langle \seta, \circ, \vone \rangle$ where:
    \begin{enumerate}[label=(\roman*)]
        \item $\seta$ is a non-empty set;
        \item $\circ$ is an associative binary operation: $\forall a, b, c \in \seta, (a \circ b) \circ c = a \circ (b \circ c)$;
        \item $\vone$ is a left and right identity element: $\forall a \in \seta, \vone \circ a = a \circ \vone = a$
        \item $\seta$ is closed under the operation $\circ$: $\forall a, b \in \seta, a \circ b \in \seta$
    \end{enumerate}
A monoid is \defn{commutative} if $\forall a,b \in \seta: a \circ b = b \circ a$. 
\end{definition}
\begin{definition}
A \defn{semiring} is a 5-tuple $\semiring = \semiringtuple$, where
    \begin{enumerate}[label=(\roman*)]
        \item $\langle \seta, \oplus, \vzero \rangle$ is a \defn{commutative monoid} over $\seta$ with identity element $\vzero$ under the \emph{addition} operation $\oplus$;
        \item $\langle \seta, \otimes, \vone \rangle$ is a \defn{monoid} over $\seta$ with identity element $\vone$ under the \emph{multiplication} operation $\otimes$;
        \item Multiplication is \defn{distributive} over addition, that is, $\forall a,b,c\in\seta$:
        \begin{itemize}
            \item $a \otimes (b\oplus c) = a \otimes b \oplus a \otimes c$;
            \item $(b\oplus c) \otimes a = b \otimes a \oplus c \otimes a$.
        \end{itemize}
        \item $\vzero$ is an \defn{annihilator} for $\seta$, that is, $\forall a \in \seta, \vzero \otimes a = a \otimes \vzero = \vzero$.
    \end{enumerate}
    A semiring is \defn{commutative} if $\langle\seta,\otimes,\vone\rangle$ is a commutative monoid.
    A semiring is \defn{idempotent} if $\forall a \in \seta: a \oplus a = a$.
\end{definition}

\begin{definition}
    A semiring $\semiring = \semiringtuple$ is \defn{complete} if it is possible to extend the addition operator $\oplus$ to infinite sums, maintaining the properties of associativity, commutativity, and distributivity from the finite case \citep[Chapter 9]{Rozenberg-1997}. 
    In this case, we can define the unary operation of the \defn{Kleene star} denoted by a superscript $\kleene{}$ as the infinite sum over powers of its operand, that is, $\forall a\in\seta$: 
    \begin{equation}
        \kleene{a} \defeq \bigoplus_{i=0}^\infty a^i 
    \end{equation}
\end{definition}
\noindent
Analogously to \cref{eq:pl}, it then follows that:
\begin{equation}
    \kleene{a} = \bigoplus_{i=0}^\infty a^i = a^0 \oplus \bigoplus_{i=1}^\infty a^i = \vone \oplus a \otimes \bigoplus_{i=0}^\infty a^i = \vone \oplus a \otimes \kleene{a} 
\end{equation}
and, similarly:
\begin{equation}
    \kleene{a} = a^0 \oplus \bigoplus_{i=1}^\infty a^i = \vone \oplus \bigoplus_{i=0}^\infty a^i \otimes a = \vone \oplus \kleene{a} \otimes a
\end{equation}

We now discuss how complete semirings can be lifted to square matrices. The definitions follow analogously to matrices over the real numbers. 
\begin{definition}
    We define \defn{semiring matrix addition} as follows. 
    Let $\matA$ and $\matB$ be $d \times d$ matrices whose entries are elements from a complete semiring $\semiring = \semiringtuple$. Then the sum ("$+$") of $\matA$ and $\matB$ is defined as:
    \begin{equation}
        (\matA \mplus \matB)_{ij} \defeq \matA_{ij} \oplus \matB_{ij} \hspace{0.5in} i,j \in 1,\dots, d
    \end{equation}
\end{definition}

\begin{definition}
    We define \defn{semiring matrix multiplication} as follows. Let $\matA$ and $\matB$ be $d \times d$ matrices whose entries are elements from a complete semiring $\semiring = \semiringtuple$. 
    Then the product of $\matA$ and $\matB$  is defined as:
    \begin{equation}
        (\matA\matB)_{ij} \defeq \bigoplus_{k=1}^d \matA_{ik} \otimes \matB_{kj} \hspace{0.5in} i,j \in 1,\dots, d
    \end{equation}
\end{definition}

We also define the \defn{zero matrix}, $\mzero$, over the complete semiring $\semiring=\semiringtuple$, such that all entries are $\vzero$, and the \defn{unit matrix} $\mone$ as $(\mone)_{ij} = \vone$ iff $i=j$ and $\vzero$ otherwise for all indices $i,j\in0,\dots,d$. 
It is then straightforward to show that matrix addition is associative and commutative, while matrix multiplication is associative and distributive over matrix addition. 
Hence, the set of square matrices of dimension $d$ over a semiring, with addition and multiplication as defined above, is itself a semiring.
Furthermore, by the element-wise definition of its addition operation, it is also complete.

\subsection{Semiring-weighted prefix algorithm}

We now consider a semiring-weighted CFG $\grammar = \langle \nonterminals, \alphabet, \start, \rules, p, \semiring \rangle$, where $\nonterminals, \alphabet, \start, \rules$ are defined as before but the weighting function $p \colon \rules \to \semiring$ now maps rules to elements of a commutative semiring $\semiring$.\footnote{We require that $\semiring$ be \emph{commutative} because the order of rule applications does not affect string weight in a weighted CFG.\looseness=-1}
Note that we no longer require the rule weights to sum to one.
As before, we define the matrix $\matP$ as the square matrix of dimension $|\nonterminals|$ whose rows and columns are indexed by the non-terminals $\nonterminals$ in some fixed order so that the entry $\matP_{ij}$ corresponds to the weight of getting the non-terminal $\nt{X_j}$ on the left after one rule application to $\nt{X_i}$. 
Since the rule weights are no longer locally normalized, however, we need to include the treesum under the right non-terminal of each rule as an additional term:\footnote{For locally normalized semirings, the treesum of any non-terminal is $\vone$. In the general case when the weights are not normalized, the treesum of a semiring-weighted WCFG can be computed through fixed point iteration in a similar way to Newton's method \citep{esparza_etal_2007_fixed_point}.\looseness=-1}
\begin{equation}
    \matP_{ij} = p(\nt{X}_i\xrightarrow{}\nt{X}_j \ \_) \defeq \bigoplus\limits_{\nt{Y}\in\nonterminals}p(\nt{X}_i\xrightarrow{}\nt{X}_j \nt{Y})\cdot \text{treesum}(\nt{Y})
\end{equation}
We can then calculate the weight of getting $\nt{X}_j$ from $\nt{X}_i$ at the leftmost non-terminal after exactly $k$ derivation steps as $(\matP^k)_{ij}$, where  $\matP^k \defeq \underbrace{\matP\cdots\matP}_{\text{$k$ times}}$. 
Finally, to get the left-corner expectations, we then need to calculate the Kleene closure over the matrix $\matP$,\footnote{Note that the Kleene closure exists since matrices with elements from a complete semiring are complete.} that is, we want to find $\kleene{\matP} = \sum_{k=0}^\infty \matP^k$. 
To compute the Kleene closure over the transition matrix we can use an efficient algorithm by \citet{lehmann-1977-algebraic} which is a generalization of the well-known shortest-path algorithm usually attributed to \citet{floyd-1962-algorithm} and \citet{warshall-1962-theorem}, but introduced previously by \citet{roy-1959-transitivie}.%
\footnote{The generalization is by way of choosing the appropriate semiring for the given problem.
By the same token, Lehmann's algorithm can also be seen as a generalization of \citeposs{kleene_1956_representation} algorithm for converting finite-state automata to regular expressions and the Gauss--Jordan algorithm for computing matrix inversion (see e.g. \citet{althoen_mclaughlin_1987_gauss_jordan}).}
The algorithm works under the condition that the Kleene closure of all individual matrix entries from semiring $\semiring$ exists, which is true for our case since we assumed $\semiring$ to be complete. 
The algorithm is shown in \cref{alg:lehmann}.

\begin{algorithm}[h!]
\caption{Lehmann's algorithm for computing the Kleene closure over a transition matrix}
\label{alg:lehmann}
\begin{algorithmic}[1]
    \Func{Lehmann($\matM$)}
    \State $d \xleftarrow{} \dim(\matM)$ \InlineComment{$\matM$ is a $d\times d$ matrix over a complete semiring}
    \State $\matM^{(0)} \xleftarrow{} \matM$
    \For{$j=1, \dots, d$}
        \For{$i=1,\dots, d$}
            \For{$k=1,\dots, d$}
                \State $\matM_{ik}^{(j)} \xleftarrow{} \matM_{ik}^{(j-1)} \oplus \matM_{ij}^{(j-1)} \otimes \kleene{\left(\matM_{jj}^{(j-1)}\right)}\otimes \matM_{jk}^{(j-1)}$
            \EndFor
        \EndFor
    \EndFor
    \State \Return $\mone + \matM^{(d)}$
    \EndFunc
\end{algorithmic}
\end{algorithm}
Lastly, since we no longer have normalized rule weights, we need to set the prefix weight of $\varepsilon$ under any non-terminal $\nt{X}\in\nonterminals$ to the treesum under $\nt{X}$.
With this, we can now generalize our prefix weight algorithm to semirings, as shown in \cref{alg:fast_semiring_lri}.
\begin{algorithm}[h!]
\caption{Faster prefix algorithm over semirings}
\label{alg:fast_semiring_lri}
\begin{algorithmic}[1]
    \State $\matP' \xleftarrow{} \text{Lehmann}(\matP)$ \InlineComment{Precompute $\kleene{\matP}$ with \cref{alg:lehmann}}\label{fsjl:pstar}
    \For {$\nt{X}_i,\nt{X}_j\in\nonterminals$}\InlineComment{Assign $\plc(\nt{X}_j\mid \nt{X}_i)$}
        \State $\plc(\nt{X}_j\mid \nt{X}_i)\xleftarrow{} {\matP'}_{ij}$
    \EndFor
    \Func{FastSemiringJL($\str=w_1 \cdots w_N,\grammar$)}
    \State $\ppre{\cdot}{\cdot}{\cdot} \xleftarrow{} \vzero$ \InlineComment{Initialize prefix probabilities}
    \For{$k \in 0,\dots, N$}
        \For{$\nt{X} \in \nonterminals$} \InlineComment{Prefix weight of $\varepsilon$}
            \State $\ppre{k}{k}{X}\xleftarrow{} \treesum(\nt{X})$ \label{fjls:eps}
        \EndFor
    \EndFor
    \State $\pinside \xleftarrow{} \text{CKY}(\str)$\InlineComment{Compute $\pinside$ with \cref{alg:cky}}
    \For {$i,j=0,\dots, N$}
        \For{$\nt{X,Z}\in\nonterminals$}\InlineComment{Compute $\gamma$ by \cref{eq:gamma}}
            \State $\gamma_{ij}(\nt{X,Z}){\xleftarrow{}}\bigoplus\limits_{\nt{Y}\in\nonterminals} p(\nt{X}{\xrightarrow{}}\nt{Y}\nt{Z})\otimes\pins{i}{j}{Y}$
        \EndFor
        \For{$\nt{X,Z}\in\nonterminals$}\InlineComment{Compute $\delta$ by \cref{eq:delta}}
            \State $\delta_{ij}(\nt{X,Z}){\xleftarrow{}}\bigoplus\limits_{\nt{Y}\in\nonterminals}\plc(\nt{Y}\mid\nt{X})\otimes\gamma_{ij}(\nt{Y},\nt{Z})$
        \EndFor
    \EndFor
    \For{$k \in 0,\dots, N-1$}
        \For {$\nt{X} \in \nonterminals$} \InlineComment{Compute base case}
            \State $\ppre{k}{k{+}1}{X}{\xleftarrow{}}\bigoplus\limits_{\nt{Y}{\in}\nonterminals}\plc(\nt{Y} \mid \nt{X})\otimes p(\nt{Y}{\xrightarrow{}}w_{k+1})$
        \EndFor
    \EndFor
    \For{$\ell \in 2 \hdots N$}
        \For{$i \in 1 \hdots N-\ell$}
            \State $k\xleftarrow{}i+\ell$
            \For{$\nt{X, Z}\in\nonterminals$} \InlineComment{Recursively compute $\pprefix$}
                \State $\ppre{i}{k}{X}$ \texttt{+=} $\bigoplus\limits_{j=i+1}^{k-1}\delta_{ij}( \nt{X},\nt{Z})\otimes\ppre{j}{k}{Z}$
            \EndFor
        \EndFor
    \EndFor
    \State\Return $\pprefix$
\EndFunc
\end{algorithmic}
\end{algorithm}

\begin{proposition}
    The semiring-weighted version of our algorithm runs in $\bigO{N^2|\nonterminals|^3+N^3|\nonterminals|^2}$.
\end{proposition}
\begin{proof}
Lehmann's algorithm, as presented in \cref{alg:lehmann}, has three nested for loops of $d$ iterations each, where $d$ is the dimension of the input matrix. 
In our case, $d$ is the number of non-terminals, $|\nonterminals|$. 
Assuming the Kleene closure of elements in $\semiring$ can be evaluated in $\bigO{1}$, this means that computing the left corner expectations in \cref{fsjl:pstar} of \cref{alg:fast_semiring_lri} takes $\bigO{|\nonterminals|^3}$, as before.
Hence, the complexity of the overall algorithm remains unchanged, that is, we can compute the prefix probabilities under a semiring-weighted, locally normalized CFG $\grammar$ in $\bigO{N^2|\nonterminals|^3+N^3|\nonterminals|^2}$.
\end{proof}

\end{document}